\documentclass[journal,draftcls,onecolumn,12pt,romanappendices,twoside]{IEEEtranTCOM}
\normalsize

\ifCLASSINFOpdf

\else

\fi

\usepackage{amsmath}
\usepackage{graphicx}
\usepackage{amssymb}
\usepackage{amsthm}
\usepackage{cite}
\usepackage{mathtools}
\usepackage{bbold}
\usepackage{steinmetz}
\usepackage{stmaryrd}
\usepackage{bbm, dsfont}
\usepackage{verbatim}
\usepackage[latin1]{inputenc}
\usepackage{tikz}
\usetikzlibrary{shapes,arrows}
\usepackage{caption}
\usepackage{enumerate}
\usepackage{color}
\renewcommand{\thesection}{\Roman{section}} 
\usepackage{steinmetz}

\usepackage{lipsum}
\usepackage[variablett]{lmodern}
\usepackage{amsbsy}
\usepackage{bm}
\usepackage{verbatim}
\usepackage[latin1]{inputenc}
\usepackage{tikz}
\usetikzlibrary{shapes,arrows}
\usepackage[font=footnotesize]{caption}
\usepackage{enumerate}
\usepackage{color}
\usepackage{subfig}
\usepackage{pgfplots}
\pgfplotsset{compat=1.10}
\usepgfplotslibrary{fillbetween}
\usetikzlibrary{calc}

\usepackage{lipsum}
\usepackage[variablett]{lmodern}
\usepackage{amsbsy}
\usepackage{bm}
\usetikzlibrary{positioning}
\usetikzlibrary{shapes,arrows}

\pgfplotsset{
	every tick label/.append style={scale=1},
	every axis/.append style={
	}
}

\newcommand{\coutone}{\mathbf{y}}

\newcommand{\chek}[1]{\quad\href{run:./references/#1}{\checkmark}}

\newcommand{\ubtw}{U_{\text{MNC}}}

\newcommand{\nlprmz}{\eta}
\newcommand{\cinzero}{\mathbf{x}}
\newcommand{\coutzero}{\mathbf{y}}

\newcommand{\cinone}{\mathbf{x}}

\newcommand{\nlprm}{\eta}

\newcommand{\cprm}{B}

\newcommand{\vrnd}{\mathbf{v}}
\newcommand{\betprm}{\kappa}
\newcommand{\noise}{\mathbf{n}}
\newcommand{\alfprm}{\zeta}
\newcommand{\lmdprm}{\lambda}

\newcommand{\gmmprm}{\mu}

\newcommand{\srnd}{\mathbf{s}}
\newcommand{\urnd}{\mathbf{t}}
\newcommand{\wrnd}{\mathbf{w}}
\newcommand{\gfun}{g}

\newcommand{\qfun}{q}

\newcommand{\lbone}{\mathrm{L}_{\mathrm{RPC}}}
\newcommand{\ubhp}{\tilde{\mathrm{U}}_{\mathrm{RPC}}}
\newcommand{\ublp}{\mathrm{U}_{\mathrm{RPC}}}

\newcommand{\red}{\mathcal{R}}                                               

\newcommand{\mf}{\mathbf}                                                

\newcommand{\one}{\mathbb{1}}					               
                         
\newcommand{\expe}{\mathrm{E}}

\newcommand{\absol}[1]{|#1|}

\newcommand{\xx}{{\mathbf{x}}}

\newcommand{\yy}{{\mathbf{y}}}

\newcommand{\be}{\begin{equation}}
\newcommand{\ee}{\end{equation}}
\newcommand{\bs}{\begin{split}}
	\newcommand{\es}{\end{split}}

\newtheorem{theorem}{Theorem}

\newcommand{\bal}{\begin{align}}
	
	\newcommand{\eal}{\end{align}}

\newcommand{\rlp}[1]{\left(#1\right)}
\newcommand{\rla}[1]{\left\{#1\right\}}
\newcommand{\rlb}[1]{\left[#1\right]}
\newcommand{\rlpo}[1]{\lefto(#1\right)}
\newcommand{\rlao}[1]{\lefto\{#1\right\}}
\newcommand{\rlbo}[1]{\lefto[#1\right]}
\newcommand{\po}{P}
\newcommand{\pn}{P_N}

\newcommand{\lefto}{\mathopen{}\left}
\newcommand{\distas}{\sim}					
\newcommand{\given}{\,\vert\,}				

\newcommand{\di}{\mathop{}\!\mathrm{d}}

\newcommand{\chone}{RPC}
\newcommand{\chtwo}{MNC}
\newcommand{\chthree}{LPC}
\newcommand{\cone}{\mathcal{C}_{\mathrm{RPC}}}
\newcommand{\ctwo}{\mathcal{C}_{\mathrm{MNC}}}
\newcommand{\cthree}{\mathcal{C}_{\mathrm{LPC}}}

\usepackage[prefix=sol]{xcolor-solarized}
\begin{document}
\title{Bounds on the {Per-Sample} Capacity of {Zero-Dispersion} Simplified Fiber-Optical Channel Models}

\author{Kamran~Keykhosravi,~\IEEEmembership{Student~Member,~IEEE,}
	Giuseppe~Durisi,~\IEEEmembership{Senior~Member,~IEEE,}\\
        and~Erik~Agrell,~\IEEEmembership{Fellow,~IEEE.}
\thanks{ This work
	was supported by the Swedish Research Council (VR) under Grant 2013-5271. }
	\thanks{This work was presented in part at the $43$th European Conference on Optical Communication (ECOC 2017). }
\thanks{The authors are with the Department of Electrical Engineering,
	Chalmers University of Technology, Gothenburg 41296, Sweden (e-mail:
	kamrank@chalmers.se; durisi@chalmers.se; agrell@chalmers.se).}
}

\maketitle

\begin{abstract}
A number of simplified models, based on perturbation theory, have been proposed for the fiber-optical channel and have been extensively used in the literature. 
Although these models are mainly developed for the low-power regime,   they are used at moderate or  high powers as well. 
It remains unclear to what extent the capacity of these models is affected by the simplifying assumptions under which they are derived. 
 { In this paper,  we  consider single-channel data transmission based on three continuous-time optical  models \textit{i)} a regular perturbative channel, \textit{ii)} a logarithmic perturbative channel, and \textit{iii)} the stochastic nonlinear  {Schr\"odinger} (NLS) channel.
We apply two simplifying assumptions on these channels to  obtain analytically tractable discrete-time models. Namely, we neglect the channel memory (fiber dispersion) and we use a sampling receiver.  These assumptions bring into question the physical relevance of the models studied in the paper. Therefore, the results should be viewed as a first step toward analyzing more realistic channels.
  We investigate the per-sample capacity of the simplified discrete-time models. Specifically,   \textit{i)} we establish tight bounds on the capacity of the regular perturbative channel; \textit{ii)} we obtain the capacity of the logarithmic perturbative channel; and \textit{iii)} we present a novel upper bound on the capacity of the zero-dispersion NLS channel.
  Our results illustrate that the capacity of these models departs from each other at high powers because these models yield  different capacity pre-logs.
	Since all three models are based on the same physical channel, our results  highlight that care must be exercised in using  simplified channel models in the high-power regime.}

\end{abstract}


\IEEEpeerreviewmaketitle

\section{Introduction}

 \IEEEPARstart{T}{he} vast majority of the global Internet traffic is conveyed through   fiber-optical networks, which form the backbone of our information society. To cope with the growing data demand, the fiber-optical  networks have evolved from  regenerated direct-detection systems to  coherent wavelength division multiplexing (WDM) ones.  Newly emerging bandwidth-hungry services, like internet-of-thing applications and cloud processing, require even higher data rates. Motivated by this ever-growing demand, an increasing attention has been devoted in recent years to the analysis of the capacity  of the  fiber-optical channel.

Finding the  capacity of  the fiber-optical channel that is governed by the stochastic nonlinear  
{Schr\"odinger} (NLS) equation \cite[Eq.~(1)]{mecozzi1994limits}, which captures the effects of Kerr nonlinearity, chromatic dispersion, and amplification noise, remains an open problem. 
An information-theoretic analysis of the NLS channel is cumbersome because of a complicated signal--noise interaction caused by the  interplay between the nonlinearity and the dispersion{\cite{essiambre_2010_jlt}. In general,  capacity analyses of optical fibers are performed either by considering simplified channels, or by evaluating mismatched decoding lower bounds \cite{merhav1994information} via simulations (see \cite{secondini2017scope} and \cite[Sec.~I]{ghozlan2017models} for  excellent literature reviews).  } 
  {Lower bounds based on the mismatch-decoding framework   go to zero after reaching a maximum  (see, for example, \cite{secondini_2013_jlt,fehenberger2015achievable,mitra_2001_nature,Ellis_2010_jlt,essiambre_2010_jlt}). Capacity lower bounds with a similar behavior are also reported in \cite{dar_2014_ol}. In \cite{splett_1993_ecoc}, it has been shown that the maximum value of a capacity lower bound can be increased by increasing  fiber dispersion, which mitigates the effects of nonlinearity.} To establish a capacity upper bound,   Kramer \textit{et al.} \cite{Kramer_2015_itw} used the split-step Fourier (SSF) method, which is  a standard approach to solve the NLS equation numerically \cite[Sec.~2.4.1]{agrawal_2007_nfo}, to derive a discrete-time channel model.  They proved that the capacity of this discrete-time model is upper-bounded by that of an equivalent AWGN channel.  In contrast to the available lower bounds, which fall to zero or saturate at high powers, this upper bound, which is the only one available for a realistic fiber channel model, grows unboundedly.

Since the information-theoretic analysis of the NLS channel is difficult, to approximate  capacity one can resort to simplified models, a number of which have been studied in the literature (see \cite{Erik_IT_friendly} and references therein for a recent review).
Two approaches to obtain such  models are to use the regular perturbation or the logarithmic perturbation methods.
 In the former,  the effects of nonlinearity are captured by   an additive perturbative term  \cite{peddanarappagari1997volterra,mecozzi_2000_ptl}.   
This approach yields a discrete-time channel with input--output relation  $\mf y_l= \mf x_l +\Delta \mf x_l+\mf n_l$ \cite[Eq.~(5)]{Erik_IT_friendly},
where  $\mf x_l$ and $\mf y_l$ are the transmitted and the received symbols, respectively; $\mf n_l$ is the amplification noise; and $\Delta \mf x_l$ is the  perturbative nonlinear distortion.
This model holds under the simplifying assumption that both  the nonlinearity and the signal--noise interaction are weak, which is reasonable only at low power.

Regular perturbative fiber-optical channel models, with or without memory, have been extensively investigated in the literature. 
 In  \cite{Meccozzi_2012_jlt}, a first-order perturbative model for WDM systems with arbitrary  filtering and sampling demodulation, and coherent detection is proposed. { The accuracy of the model is assessed by comparing the value of a mismatch-decoding lower bound, which is derived  analytically based on the perturbative model, with simulation results over a realistic fiber-optical channel.} {A good agreement  at all power levels is observed.}
 The capacity of a  perturbative multiple-access channel is studied in \cite{taghavi_2006_tinf}.  It is shown that the nonlinear crosstalk between  channels does not affect the capacity region when the information from all the channels is optimally used at each detector.
  However, if joint processing is not possible { (it is typically computationally demanding \cite{dar2014accumulation})}, the channel capacity is limited by the inter-channel distortion.

Another class of simplified models, which are equivalent to the regular perturbative ones up to a first-order  linearization, is that of logarithmic perturbative models, where the nonlinear distortion term $\Delta \mf x_l$ is modeled as a phase shift. This yields a discrete-time channel with input--output relation $\mf y_l= \mf x_l e^{j\Delta \mf x_l}+\mf n_l$ \cite[Eq.~(7)]{Erik_IT_friendly}.
{ In \cite{ghozlan2017models}, a single-span optical channel model for a two-user WDM transmission system is developed from a coupled NLS equation, neglecting the dispersion effects within the WDM bands. The channel model in \cite{ghozlan2017models} resembles the  perturbative logarithmic models. The authors study the capacity region of this channel in the high-power regime. 
	It is shown that the capacity pre-log pair  (1,1) is achievable, where the capacity pre-log is defined as the asymptotic limit of $\mathcal{C}/\log P$ for $P\to\infty$, where $P$ is the input power and $\mathcal{C}$ is the  capacity.}

Despite the fact that the aforementioned simplified channels  are valid in the low-power regime,  these  models are often used also in the  moderate- and high-power regimes.
Currently, it is unclear to what extent the simplifications used to obtain these models  influence the capacity at high powers. 
To find out, we study the capacity of two { single-channel} memoryless perturbative  models, namely, a \textit{regular perturbative channel (\chone)},  and a \textit{logarithmic perturbative channel (\chthree)}.
To assess  accuracy of these two perturbative models, we investigate also the { per-sample} capacity of a \textit{memoryless NLS channel (\chtwo)}.

{ The analysis in this paper suffers from two shortcomings. First, the channel memory is ignored, i.e., the dispersion is set to zero. Second, a sampling receiver is used to obtain discrete-time models from continuous-time channels. 
	These two assumptions were first applied to the NLS equation in \cite{mecozzi1994limits} to obtain  an analytically tractable channel model. This channel model was developed also in  \cite{ho_2005_book,turitsyn_2003_prl,yousefi_2011_tinf} using different methods. In this paper, we  refer to this model as \chtwo{}. 
	Zero-dispersive channels do not model correctly the actual optical fiber and are used in the literature to perform theoretical analyses that are still out of reach in the dispersive case because of complexity.
	The sampling receiver also does not capture the effects of spectral broadening and neglects temporal correlation of the received signal. Hence, it is suboptimal. These shortcomings of the sampling receiver are elaborated upon in \cite{kramer2017information}, where the capacity of a nondispersive NLS channel with a band-limited receiver is upper-bounded. 
	Deploying the two above-mentioned  assumptions is a common approach to enable information-theoretic analyses of the fiber-optical channel. The results obtained in this way should only be considered as a first step towards investigating  more realistic channel models.      
	}

In \cite{turitsyn_2003_prl}, a lower bound on the per-sample capacity of the memoryless NLS  channel is derived,  which proves that the  capacity goes to infinity with power.
In \cite{yousefi_2011_tinf}, the capacity of the same channel is evaluated  numerically. Furthermore, it is shown that the capacity pre-log  is $1/2$. The only known nonasymptotic 
upper bound on the capacity of this channel is $\log(1+\mathrm{SNR})$ (bits per channel use) \cite{Kramer_2015_itw}, where $\mathrm{SNR}$ is the signal-to-noise ratio. This upper bound  holds also for the general case of nonzero dispersion.

The novel contributions of this paper are as follows. 
 First,  we tightly bound the capacity of the \chone{} model and prove that its capacity pre-log is 3.
Second, the capacity of the \chthree{} is readily shown to be the same as that of an AWGN channel with the same input and noise power. Hence, the capacity pre-log of the \chthree{} is $1$.
Third, we establish a novel upper bound\footnote{This upper bound was first presented in the conference version of this manuscript \cite{keykhosravi_ecoc_17}.} on the capacity of the \chtwo.  
Our upper bound improves the previously known upper bound \cite{Kramer_2015_itw} on the capacity of this channel significantly and together with the a proposed lower bound allows one to characterize the capacity of the \chtwo{} accurately.

Although all three  models represent the same physical optical channel, their capacities behave very differently in the high-power regime. 
This result highlights the profound impact of the simplifying assumptions on the capacity at high powers, and indicates that care should be taken in translating the results obtained based on these models into   guidelines for system design.

The rest of this paper is organized as follows. In Section~\ref{s2}, we  introduce  the three channel models. 
In Section~\ref{s3}, we present  upper and lower bounds on the capacity of these channels and establish the capacity pre-log  of the perturbative models.
   Numerical results are provided in Section~\ref{sec:num:example}.
    We conclude the paper in Section~\ref{sec:conc}. The proofs of all theorems are given in the appendices. 

\paragraph*{Notation}
Random quantities are denoted by boldface letters.
We use $\mathcal{CN}\rlpo{0,\sigma^2}$ to denote the complex zero-mean circularly symmetric Gaussian distribution with variance $\sigma^2$.  We write $\red\lefto( x\right)$,  $\absol{ x}$, and $\phase{x}$ to denote the real part, the absolute value, and the phase of a complex number $x$.
All  logarithms are in base two. The mutual information between two random variables $\mf{x}$ and $\mf{y}$ is denoted by $I(\mf{x};\mf{y})$. The entropy  and differential entropy are denoted by $H(\cdot)$ and $h(\cdot)$, respectively.   Finally, we use $\one(\cdot)$ for the indicator function, and $*$ for the convolution operator.

\section{Channel Models}\label{s2}
The fiber-optical  channel is well-modeled by the NLS equation, which describes  the propagation of a complex baseband electromagnetic field  through a lossy single-mode fiber  as 
\be
\frac{\partial  \mf a}{\partial z}+\frac{\alpha-g}{2}\mf a+j\frac{\beta_2}{2}\frac{\partial ^2 \mf  a}{\partial t^2}-j\gamma | \mf a|^2  \mf a= \mf n.\label{nlse}
\ee
Here, $\mf a =\mathbf{a}(z , t)$ is the complex baseband signal  at time $t$ and location $z$. The parameter $\gamma$ is the nonlinear coefficient, $\beta_2$ is the group-velocity dispersion parameter,  $\alpha$  is the attenuation constant, $g=g(z)$ is the  gain profile of the amplifier, and $\mf n=\mf n(z, t)$ is the  Gaussian amplification noise, { which is bandlimited because of the inline channel filters.} The third term on the left-hand side of \eqref{nlse} is responsible for the channel memory and the fourth term for the channel nonlinearity.  

To compensate for the fiber losses, two types of signal amplification can be deployed, namely, distributed and lumped amplification.  The former method compensates for the fiber loss continuously along the fiber, whereas the latter method boosts the signal power by dividing the fiber into several spans and using an optical amplifier at the end of each span. With distributed amplification, which we  focus on in  this paper, the noise can be described by the  autocorrelation function{\cite{essiambre_2010_jlt}}
\begin{equation}\label{eq:autocorolation}
\expe\rlbo{\mf n(z,t)\mf n^*(z',t')}=\alpha n_{\mathrm{sp}} h \nu \delta_{{W_N}}(t-t')\delta(z-z').
\end{equation}  
Here, $n_{\mathrm{sp}}$  is the spontaneous emission factor, $h$ is  Planck's constant, and $\nu$ is the optical carrier frequency. Also,   $\delta(\cdot)$ is the Dirac delta function and $\delta_{{W_N}}(x)=W_N \mathrm{sinc}(W_Nx)$, where $W_N$ is the noise bandwidth.
In this paper, we shall focus on the ideal distributed-amplification case $g(z)=\alpha$.

{We use a sampling receiver to go from continuous-time channels to  discrete-time ones. A comprehensive description of the sampling receiver and of the induced discrete-time channel is provided in \cite[Section~III]{yousefi_2011_tinf}. Here, we review some crucial elements of this description  for completeness.	Assume that a signal $\mf a(0,t)$, which is band-limited to $W_0$ hertz,  is transmitted through a zero-dispersion NLS channel ((1) with $\beta_2=0$) in the time interval $[0 , \mathcal{T}]$. Because of nonlinearity, the bandwidth of the received signal $\mf a(L,t)$ may be larger than that of $\mf a(0,t)$. To avoid signal distortion by the inline filters, we assume that $W_0$ is set such that $\mf a(z,t)$ is band-limited to $W_N$ hertz for $0\leq z\leq L$. Since $W_0\leq W_N$,  assuming $W_N\mathcal{T}\gg 1$, both the transmitted and the received signal can be represented by $2W_N\mathcal{T}$ equispaced samples. The transmitter encodes  data into  subsets of these samples of cardinality $2W_0\mathcal{T}$, referred to as the principal samples. At the receiver,  demodulation is performed by sampling $\mf a(L,t)$ at instances corresponding to the principal samples. This results in $2W_0\mathcal{T}$ parallel independent discrete-time channels that have the same input--output relation.

	The sampling receiver has a number of shortcomings \cite{kramer2017information} and using it should be considered  a simplification. The resulting discrete-time model is used extensively in the literature (see for example \cite{yousefi_2011_tinf,mecozzi1994limits,ho_2005_book,turitsyn_2003_prl,lau_2007_jlt,tavana_ecoc_18}), since it makes analytical calculation possible.   In this paper, we apply the sampling receiver not only to the memoryless NLS channel but also to the memoryless perturbative models. 
	}

 { Next,  we review two perturbative channel models that are used in the literature to approximate the solution of the NLS equation \eqref{nlse}.
 Among the multiple variations of  perturbative models available in the literature, we use the ones proposed in \cite{forestieri2005solving}.
 For both perturbative models,  first    continuous-time  dispersive   models are introduced, and then  memoryless discrete-time channels are developed by assuming that  $\beta_2=0$ and  by using a sampling receiver.
 Finally,  we introduce the  \chtwo{} model, which is derived from \eqref{nlse} under the two above-mentioned assumptions.
 }

{
\paragraph*{Regular perturbative channel (\chone) }
Let $\mf a^{\mathrm{li}}(z,t)$ be the solution of the linear noiseless NLS equation (Eq.~\eqref{nlse} with $\mathbf{n}(z,t)=0$ and $\gamma=0$). It can be computed as $\mf a^{\mathrm{li}}(z,t)=\mf a(0,t)*h(z,t)$, where $h(z,t)=\mathcal{F}^{-1}\rla{\exp\rlpo{j\beta_2\omega^2z/2}}$ and $\mathcal{F}^{-1}(\cdot)$ denotes the inverse Fourier transform. 
%
In the regular perturbation method,  the output of the NLS channel \eqref{nlse} is approximated as \cite[Eq.~(5)]{Erik_IT_friendly}
\begin{equation}\label{eq12}
\mf a(L,t)=\mf a^{\mathrm{li}}(L,t)+\Delta\mf a(L,t)+\mf w(L,t).
\end{equation}
Here, $L$ is the fiber length,  $\Delta\mf a(z,t)$ is the nonlinear perturbation term, and 
\begin{equation}\label{eqw}
	\mf w(L,t)=\int_0^L \mf n(z , t) \di z
	\end{equation}
 is the accumulated amplification noise. The  first-order approximation of $\Delta\mf a(L,t)$ is \cite[Eq.~(13)]{forestieri2005solving}
 \begin{equation}\label{eq11}
\Delta\mf a(L,t)=j\gamma\int\limits_0^L\rlb{\left|\mf a^{\mathrm{li}}(\zeta,t)\right|^2\mf a^{\mathrm{li}}(\zeta,t)}*h(L-\zeta,t)\di \zeta
 \end{equation}
 where the convolution is over the time variable.
 Neglecting dispersion (i.e., setting $\beta_2=0$), we have $h(z,t)=\delta(t)$ and $\mf a^{\mathrm{li}}(\zeta,t)= \mf a(0,t)$. Using this in  \eqref{eq11}, and then substituting \eqref{eq11} into \eqref{eq12}, we obtain 
 \begin{equation}\label{eq13}
 \mf a(L,t)=\mf a(0,t)+jL\gamma|\mf a(0,t)|^2\mf a(0,t)+\mf w(L,t).
 \end{equation}
 %
 %
%
%
%
%
Finally, by deploying sampling receiver, we obtain the  discrete-time channel model
\begin{IEEEeqnarray}{rCl}\label{eq:channel1}
	\mf y=\mf x + j\nlprm|\mf x|^2\mf x +\mf n.
\end{IEEEeqnarray}
Here, $\mf n\distas \mathcal{CN}\rlpo{0,\pn}$, 
  \begin{equation}\label{erik}
\pn=2\alpha n_{\mathrm{sp}} h \nu L W_N
  \end{equation}
   is the total noise power, and  
\begin{IEEEeqnarray}{rCl}\label{eq:nlprm}
	\nlprm=\gamma L.
\end{IEEEeqnarray}
We refer to  \eqref{eq:channel1} as the \chone.

\paragraph*{Logarithmic perturbative channel (\chthree)}
Another method for approximating the solution of the NLS equation \eqref{nlse} is to use logarithmic perturbation. With this method, the output signal is approximated  as \cite[Eq.~(7)]{Erik_IT_friendly}
\begin{equation}\label{eq:LP}
\mf a(L,t)=\mf a(0,t)\exp\rlpo{j\Delta\bm\theta(L,t)}+\mf w(L,t)
\end{equation}
where $\mf w(L,t)$ is the same noise term as in \eqref{eq12}--\eqref{eqw}. The first-order approximation of $\Delta\bm\theta(L,t)$ is \cite[Eq.~(19)]{forestieri2005solving}
\begin{equation}\label{eq15}
\Delta\bm\theta(L,t)=\frac{\gamma}{\mf a^{\mathrm{li}}(L,t)}\int\limits_0^L \rlb{\left|\mf a^{\mathrm{li}}(\zeta,t)\right|^2\mf a^{\mathrm{li}}(\zeta,t)}*h(L-\zeta,t)\di \zeta.
\end{equation}
Under the zero-dispersion assumption ($\beta_2=0$), we have $h(z,t)=\delta(t)$ and $\mf a^{\mathrm{li}}(\zeta,t)= \mf a(0,t)$. Using this in  \eqref{eq15}, and then substituting \eqref{eq15} into \eqref{eq:LP}, we obtain
 \begin{equation}\label{eq14}
 \mf a(L,t)=\mf a(0,t)e^{j\gamma L|\mf a(0,t)|^2}+\mf w(L,t).
 \end{equation}
Finally, by sampling the output signal, the  discrete-time channel 
\begin{IEEEeqnarray}{c}\label{eq:channel:one:step}
	\coutzero=\cinzero e^{j\nlprmz|\cinzero|^2}+\noise
\end{IEEEeqnarray}
is obtained, where $\noise\distas\mathcal{CN}\rlpo{0,\pn}$, $\pn$ is given in \eqref{erik}, and $\nlprm$ is defined in \eqref{eq:nlprm}.
We note that the channels \eqref{eq:channel1} and \eqref{eq:channel:one:step} are equal up to a first-order linearization, which is accurate in the low-power regime.
Furthermore, one may also obtain the model in \eqref{eq14}  by solving \eqref{nlse} for $\beta_2=0$, $\mf n=0$, and $g=\alpha$ and by adding the noise at the receiver.

\paragraph*{Memoryless NLS Channel (\chtwo)}
Here, we shall  study the underlying  NLS channel in \eqref{nlse} under the assumptions that $\beta_2=0$ and that  a sampling receiver is used to obtain a discrete-time channel.
Let $\mf r_0$ and $\bm \theta_0$ be the amplitude and the phase of a transmitted symbol $\mf x$, and let $\mf r$ and $\bm \theta$ be those of the  received samples $\mf y$.
The  discrete-time channel input--output relation can be described by the conditional probability density function (pdf) \cite[Ch.~5]{ho_2005_book} (see also \cite[Sec.~II]{lau_2007_jlt})
\begin{equation}\label{cd:pdf}
f_{\mf r,\bm \theta\mid \mf r_0, \bm \theta_0}(r,\theta|r_0,\theta_0)=\frac{f_{\mf r\mid \mf r_0}(r\mid r_0)}{2\pi}+\frac{1}{\pi}\sum_{m=1}^{\infty} \red\rlpo{C_m(r)e^{-jm(\theta-\theta_0)}}.
\end{equation}
The conditional pdf  $f_{\mf r\mid \mf r_0}(r\mid r_0)$ and the Fourier coefficients $C_m(r)$ in \eqref{cd:pdf} are given by
\begin{align}
f_{\mf r\mid \mf r_0}(r\mid r_0)&=\frac{2r}{P_N}\exp\rlpo{-\frac{r^2+r_0^2}{P_N}}I_0\rlpo{\frac{2rr_0}{P_N}}\\
C_m(r)&=2r\nu_m\exp\rlpo{-\rlp{r^2+r_0^2}\nu_m\cos x_m}I_m\rlpo{2rr_0\nu_m}.
\end{align}
Here, $I_m(\cdot)$ denotes the $m$th order modified Bessel function of the first kind, and\footnote{The complex square root in \eqref{eq17} is a two-valued function, but both choices give the same values of $\nu_m$ and $C_m(r)$.}
\begin{align}
x_m&=\rlp{\frac{2jm\gamma r_0^2\pn L}{2r_0^2+\pn}}^{1/2}\label{eq17}\\
\nu_m&=\frac{x_m}{\pn \sin x_m}.
\end{align}


In the next section, we study the capacity of the channel models given in \eqref{eq:channel1}, \eqref{eq:channel:one:step}, and~\eqref{cd:pdf}. Since all of these models are memoryless, their capacities under a power constraint $P$ are given by
\begin{align}
 \mathcal C=\sup \  &I(\xx;\yy)\label{capacity}
\end{align}
 where the supremum is  over all complex probability distributions of  $\mf x$ that satisfy the average-power constraint 
\begin{equation}
\expe\rlbo{|\mf x|^2}\leq P.\label{eq:power:constraint}
\end{equation}
}
\section{ Analytical Results}\label{s3}

In this section, we study the capacity of the \chone, the \chthree, and the \chtwo{} models. All these  models  are based on the same fiber-optical channel and share the same set of parameters. Bounds on the capacity of the  \chone{} in \eqref{eq:channel1} are provided in  Theorems~\ref{thm:lower_x3}--\ref{thm:x3:ubn:lp}.
 Specifically, in Theorem~\ref{thm:lower_x3} we establish a closed-form lower bound, which, together with the upper bound provided in Theorem~\ref{thm:ub:x3}, tightly bounds  capacity (see Section~\ref{sec:num:example}).
A different upper bound is provided  in Theorem~\ref{thm:x3:ubn:lp}. Numerical evidence suggests that this alternative bound is less tight than the one provided in Theorem~\ref{thm:ub:x3} (see Section~\ref{sec:num:example}). However, this alternative bound has a simple analytical  form, which makes it easier to characterize it asymptotically.
 By using the bounds derived in Theorem~\ref{thm:lower_x3} and Theorem~\ref{thm:x3:ubn:lp}, we prove that the capacity pre-log  of the \chone{} is  $3$.
 In Theorem~\ref{thm:onestep}, we present for completeness the (rather trivial) observation that the capacity of the \chthree{} in \eqref{eq:channel:one:step} coincides with that of an equivalent AWGN channel. Hence, the capacity pre-log is $1$.
Finally, in Theorem~\ref{thm:ub:nl:ssfm},  we provide  an upper bound on the capacity of the \chtwo{} in \eqref{cd:pdf}, which  improves the previous known upper bound \cite{Kramer_2015_itw} significantly, and, together with a proposed  capacity lower bound, yields a tight characterization of  capacity (see Section~\ref{sec:num:example}).

\subsection{Capacity Analysis of the \chone{}} 
\begin{theorem}\label{thm:lower_x3}\normalfont
The capacity  $\cone$ of the \chone{} in \eqref{eq:channel1}   is lower-bounded by 
\begin{align}
\cone\geq\lbone(\po)&=\max_{\lambda}\rlao{\log\rlpo{\frac{\lmdprm^2+6\nlprm^2}{\lmdprm^3  \pn}\ e^{\frac{12\nlprm^2}{\lmdprm^2+6\nlprm^2}}+1}} \label{eq:ch1:lb}
\end{align}
where $\lmdprm$ is positive and satisfies the  constraint
\begin{align}
 	\frac{18\nlprm^2+\lmdprm^2}{\lmdprm\rlp{6\nlprm^2+\lmdprm^2}}\leq\po.\label{eq:powr:cons}
\end{align}
 Furthermore, the maximum in \eqref{eq:ch1:lb} is achieved by the  unique real solution of the  equation
\begin{align}
\po \lmdprm^3&-\lmdprm^2 +6\po\nlprm^2\lmdprm-18\nlprm^2=0.\label{eq:lmd:tm}
\end{align}

\end{theorem}
\begin{proof}
See Appendix~\ref{app:prf:tm:ch1:lb}.
\end{proof}

\begin{theorem}\label{thm:ub:x3}\normalfont
	The capacity of the \chone{} in \eqref{eq:channel1}  is upper-bounded by 
	\begin{align}
	\cone&\leq\ublp(\po)\nonumber\\\label{eq:thm:ub:x3:hp}
		 &=\min_{\gmmprm>0,\ \lambda>0}\lefto\{\log{\frac{\gmmprm^2+6\nlprm^2}{\gmmprm^3e\pn}}+\lambda+ \max_{s>0}\rlao{\gmmprm \expe\rlbo{\qfun\rlpo{|\coutone|^2}\mid |\cinone|^2=s}\log e-\lambda\frac{s+\pn}{P+\pn}}\right\}.
	\end{align} 
Here, $\qfun(x)=g^{-1}(x)$, where $\gfun\rlpo{x}=x+\nlprm^2x^3$. 
\end{theorem}
\begin{proof}
See Appendix~\ref{app:proof:thm:ub:x3}.
\end{proof}
Note that, given $|\xx|^2=s$, the random variable 
$
2|\coutone|^2/\pn
$ 
is conditionally distributed as a noncentral chi-squared random variable with $2$ degrees of freedom and noncentrality parameter $2(s+\nlprm^2s^3)/\pn$. This enables  numerical computation of $\ublp(\po)$.

\begin{theorem}\label{thm:x3:ubn:lp}\normalfont
The capacity of the \chone{} in \eqref{eq:channel1}  is upper-bounded by 
\begin{align}
	\cone\leq\ubhp(\po)=\label{eq:thm:ub:x3:hp2}
	\min_{\gmmprm>0}\rlao{\log\rlpo{\frac{\gmmprm^2+6\nlprm^2}{\gmmprm^3 e \pn} }+\gmmprm\rlpo{\po+\cprm}\log e}
\end{align} 
where
\begin{align}
	\cprm=\pn+\frac{\sqrt{\pi\pn}}{12^{3/8}\sqrt{(\sqrt{3}-1)\nlprm}}.
\end{align}
Furthermore, the minimum in \eqref{eq:thm:ub:x3:hp2} is achieved by the unique real solution of the equation  
\begin{align}
	&\rlp{\po+\cprm}\gmmprm^3-\gmmprm^2+6\nlprm^2\rlpo{\po+\cprm}\gmmprm-18\nlprm^2=0.\label{eq:gmmprm:def}
\end{align}

\end{theorem}
\begin{proof}
	See Appendix~\ref{app:proof:ub:x3:lp}.
\end{proof}

\paragraph*{Pre-log analysis}
 By substituting $\gmmprm=1/P$ into \eqref{eq:thm:ub:x3:hp2}, we see that 
\begin{IEEEeqnarray}{c}
	\lim\limits_{\po\to\infty}\rlb{\cone-3\log(\po)}\leq \log\rlpo{\frac{6\nlprm^2}{\pn}}.\label{eq:ub:pl:ch1}
\end{IEEEeqnarray}
Furthermore, since
\begin{IEEEeqnarray}{rCl}
	\frac{18\nlprm^2+\lmdprm^2}{\lmdprm\rlp{6\nlprm^2+\lmdprm^2}}&\leq&\frac{18\nlprm^2+3\lmdprm^2}{\lmdprm\rlp{6\nlprm^2+\lmdprm^2}}\label{eq:pwchkr31}\\
	&=&\frac{3}{\lambda}
\end{IEEEeqnarray}
we can obtain a valid lower bound on $\cone{}$ by substituting  $\lmdprm=3/\po$ into \eqref{eq:ch1:lb}. Doing so, we obtain
\begin{align}
	\lim\limits_{\po\to\infty}\rlb{\cone-3\log(\po)}\geq\log\rlpo{\frac{2\nlprm^2e^2}{9\pn}}.\label{eq:lb:pl:ch1}
\end{align}
It follows from \eqref{eq:ub:pl:ch1} and \eqref{eq:lb:pl:ch1}  that the capacity pre-log of the \chone{} is 3. 

\subsection{Capacity Analysis of the \chthree{}}

\begin{theorem}\label{thm:onestep}\normalfont
	The capacity of the  \chthree{} in \eqref{eq:channel:one:step} is
	\begin{IEEEeqnarray}{c}\label{eq:AWGN}
		\cthree=\log\rlpo{1+\frac{\po}{\pn}}.
	\end{IEEEeqnarray}
	
	\begin{proof}
		We use the maximum differential entropy lemma \cite[Sec.~2.2]{elgamal_2011_net} to  upper-bound $\cthree$ by $\log\rlpo{1+{\po}/{\pn}}$. Then, we  note that we can achieve this upper bound by choosing $
		\cinzero\distas\mathcal{CN}\rlpo{0,\po}
		$.
		
	\end{proof}
\end{theorem}

\subsection{Capacity Analysis of the \chtwo{}}
A novel upper bound on the capacity of the \chtwo{} in \eqref{cd:pdf}
 is presented in the following theorem \cite{keykhosravi_ecoc_17}.
\begin{theorem}\label{thm:ub:nl:ssfm}\normalfont
The capacity of the \chtwo{} in \eqref{cd:pdf} is upper-bounded by 
\begin{IEEEeqnarray}{rCl}\label{eq:thm:ub:nl:ssfm}
	 \ctwo &\leq& \ubtw (P)\\
	 &=&\min_{\lambda>0,\ \alpha>0} \bigg\{\alpha\log\rlpo{\frac{\po+\pn}{\alpha}}+\log\lefto(\pi\Gamma(\alpha)\right)+\lambda+\max_{r_0>0}\{g_{\lambda,\alpha}(r_0, P)\}\bigg\}
\end{IEEEeqnarray}
where  $\Gamma(\cdot)$ denotes the Gamma function and
\begin{align}\label{gdef}
	g_{\lambda,\alpha}(r_0, P)&=({\alpha\log e}-\lambda)\frac{ r_0^2+\pn}{P+\pn}
	+(1-2\alpha)\expe\big[\log(\mf r)\mid \mf r_0=r_0\big]\nonumber\\
	&\qquad-h\lefto(\mf r\mid \mf r_0=r_0\right) -h(\bm \theta\mid \mf r, \mf r_0=r_0, \bm \theta_0=0).
\end{align}
The upper bound $\ubtw (P)$ can be calculated numerically using the expression for the conditional pdf $f_{\mf r,\bm \theta\mid \mf r_0, \bm \theta_0}(r,\theta|r_0,\theta_0)$ given in \eqref{cd:pdf}.
\end{theorem}
\begin{proof}
See Appendix~\ref{app:proof:ub:nl:ssfm}.
\end{proof}
\begin{table}[!t]
	\renewcommand{\arraystretch}{1.3}
	\caption{Channel parameters.}
	\label{table1}
	\centering
	\begin{tabular}{c c c }
		\hline
		\hline
		Parameter&Symbol& Value\\
		\hline
		Attenuation &$\alpha$ & $0.2 \ \mathrm{dB/km}$\\
		Nonlinearity& $\gamma$ & $1.27 \ \rlp{\mathrm{W \cdot km}}^{-1}$\\
		Fiber length &$L$ & $5000 \ \mathrm{km}$\\
		Maximum bandwidth&$W_N$&$125 \ \mathrm{GHz}$\\
		Emission factor&$n_{\mathrm{sp}}$&$1$\\
		Photon energy &$h\nu$ & $1.28\cdot10^{-19}J$\\ 
		Noise variance&$\pn$&$-21.3\ \mathrm{dBm}$\\
		\hline
		\hline
	\end{tabular}
\end{table}

\section{Numerical Examples}\label{sec:num:example}
In Fig.~\ref{fig_plot_bounds}, we evaluate  the bounds derived in Section~\ref{s3} for a fiber-optical channel whose parameters are listed in Table~\ref{table1}.\footnote{ The channel parameters are the same as in \cite[Table~I]{yousefi_2011_tinf}.}
 Using \eqref{eq:nlprm}, we obtain $\nlprm=6350 \ \text{W}^{-1}$.

As can be seen from Fig.~\ref{fig_plot_bounds}, the capacity of the \chone{} is tightly bounded between the upper bound $\ublp(P)$ in \eqref{eq:thm:ub:x3:hp} and the lower bound $\lbone(P)$ in \eqref{eq:ch1:lb}. Furthermore, one can observe that although the alternative upper bound $\ubhp(P)$ in \eqref{eq:thm:ub:x3:hp2} is loose at low powers, it becomes tight in the moderate- and high-power regimes. 

\begin{figure}[t]
	\centering
	\includegraphics{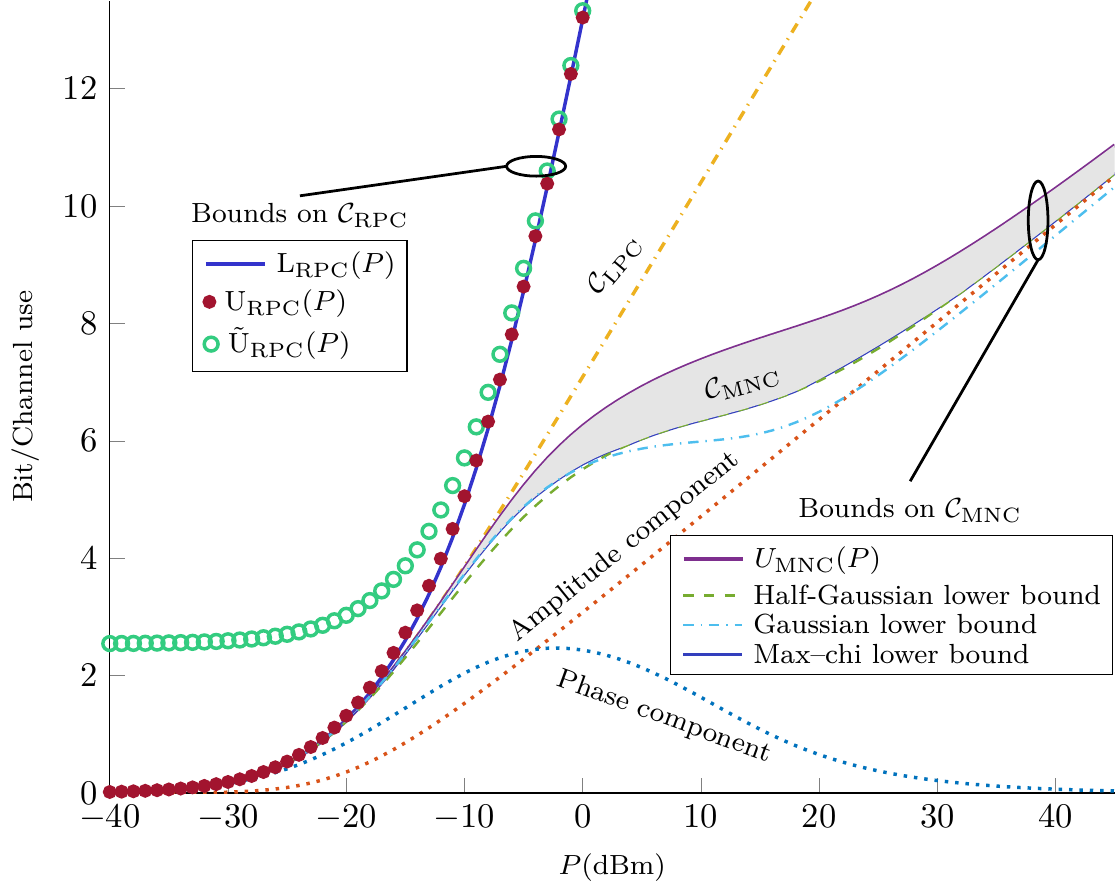}
	\caption{  Capacity bounds for the \chone{} in \eqref{eq:channel1} and the \chtwo{} in \eqref{cd:pdf}, together with the capacity of the \chthree{} in \eqref{eq:channel:one:step}. {The amplitude and the phase components of the  half-Gaussian lower bound for the \chtwo{} are also plotted.}
	} 
	\label{fig_plot_bounds}
\end{figure}

 We also plot the upper bound $\ubtw(P)$ on the capacity of the \chtwo{}.
  It can be seen that $\ubtw(P)$ improves substantially on the  upper bound given in \cite{Kramer_2015_itw}, i.e., the capacity of the corresponding AWGN channel \eqref{eq:AWGN} (which coincides with $\cthree{}$). 
{ As a lower bound on the \chtwo{} capacity, we propose the mutual information in \eqref{capacity} with an input $\mathbf x$ with uniform phase and amplitude $\mathbf r_0$ following a chi distribution with $k$ degrees of freedom. Specifically, we set
	\begin{equation}\label{eq:hg}
	f_{\mathbf r_0}(r_0)=\frac{2r_0^{k-1}}{\Gamma(k/2)}\rlp{\frac{k}{2P}}^{k/2}\!\!\!\!\exp\rlpo{\frac{-kr_0^2}{2P}}    
	\end{equation}
	 where $\Gamma(\cdot)$ denotes the gamma function.
	 The parameter $k$ is optimized for each power\footnote{Due to the computational complexity, we only considered $k$ values from $0.5$ to $2.5$ in steps of $0.5$.}. We calculated the bound numerically and include it in Fig.~\ref{fig_plot_bounds} (referred to as max--chi lower bound). We also include two lower bounds corresponding to $k=1$ (with half-Gaussian amplitude distribution, first presented in \cite{yousefi_2011_tinf}) and $k=2$ (with Rayleigh-distributed amplitude, or equivalently, a complex Gaussian input $\mathbf x$, first presented in \cite{keykhosravi_ecoc_17}).
The max--chi lower bound coincides with these two lower bounds  at asymptotically low and high power, and improves slightly thereon at intermediate powers (around $0$ dBm), similarly to the numerical bound in \cite{ShenLi_ecoc_18}. Specifically, at asymptotically low powers, $k = 2$ (Gaussian lower bound) is optimal. This is expected, since the channel is essentially linear at low powers. At high powers, on the other hand, the optimal $k$ value approaches $1$ (half-Gaussian lower bound), which is consistent with \cite{yousefi_2011_tinf},
where it has been shown that half-Gaussian amplitude distribution is capacity-achieving for the MNC in the high-power regime. Finally, we observed that  $k=0.5$ is optimal in the power range $18\leq P\leq 32$ dBm.  }


Fig.~\ref{fig_plot_bounds} suggests that $\ctwo{}$ experiences  changes in  slope at about $0$ and $30$ dBm {(corresponding to the inflection points at about $-10$ dBm and $20$ dBm)}.
{ To explain this behavior, we evaluate the phase and the amplitude components of the half-Gaussian lower bound. Specifically, we split the mutual information  into two parts as
	\begin{align}
I(\mf x ; \mf y)&=	I(\mf r_0,\bm \theta_0 ; \mf r, \bm \theta )\\
&=I(\mf r_0,\bm \theta_0;\mf r)+I(\mf r_0 ,\bm \theta_0;\bm \theta\mid \mf r).\label{split}
	\end{align}
	The first term in \eqref{split} is the amplitude component and the second term is the phase component of the mutual information. These two components are evaluated for the half-Gaussian amplitude distribution  and plotted in Fig.~\ref{fig_plot_bounds}. It can be seen from Fig.~\ref{fig_plot_bounds} that the  amplitude component is monotonically increasing with power while the phase component goes to zero with power after reaching a maximum. Indeed, 	  at high powers the phase of the received signal becomes uniformly distributed over $[0 ,  2\pi]$  and independent of the transmitted signal \cite[Lem.~5]{yousefi2016asymptotic}.   	By adding these two components one obtains a capacity lower bound  that changes concavity at two points. }

As a final observation, we note  that  $\cone$ diverges from $\ctwo$ at about $-15$ dBm, whereas  $\cthree{}$ diverges from $\ctwo{}$ at about $-5$ dBm. Since the \chtwo{}  describes the nondispersive NLS channel {more accurately than the other two channels,} this result suggests that the perturbative models are grossly inaccurate in the high-power regime.

\section{Discussion and Conclusion}\label{sec:conc}
The capacity of three   { single-channel} optical  models, namely, the \chone{}, the \chthree{}, and the \chtwo{} were investigated. 
{ All three models are developed under two simplifying assumptions: channel memory is ignored and a sampling receiver is applied.
   Furthermore, two of these models, i.e., the \chone{} and the \chthree, are based on  perturbation theory and ignore signal--noise interaction,  which makes them accurate only in the low-power regime. 
   By tightly bounding the capacity of the \chone{}, by characterizing the capacity of the \chthree, and by developing a tight upper bound on the capacity of the \chtwo, we showed that the capacity of these models, for the same underlying physical channel, behave very differently at high powers. Since the \chtwo{} is a more accurate channel model than the other two, one may conclude that the perturbative models become grossly inaccurate at high powers in terms of capacity calculation. 

   	Note that the \chthree{} model can be obtained from the \chtwo{} by neglecting the signal--noise interaction.
   	Comparing the capacity of these two channels allows us to conclude that the impact of neglecting the signal--noise interaction on capacity is significant.
   	Observe also that the capacity of the \chthree{} model grows quickly with power, because of the large capacity pre-log. Such a behavior is caused by the additive model used for the nonlinear distortion, which causes an artificial power increase at high SNR.
   	A more accurate model than the \chone{} may be obtained by performing a normalization that conserves the signal power.  Future work should consider more realistic channel models with nonzero dispersion and with  more practical receivers.  Zero-dispersion models do not represent the physical fiber-optical channel and the analytical results based on these models should serve only as a first step towards the study of more physically-relevant channels.       }

\appendices

\section{Proof of Theorem~\ref{thm:lower_x3}}\label{app:prf:tm:ch1:lb}

The capacity of the regular perturbative  channel can be written as
\begin{IEEEeqnarray}{s?l}\label{eq:shannon}
\ &\cone= \sup I\rlpo{\cinone; \coutone}
\end{IEEEeqnarray}
where the supremum is  over all the probability distributions on $\mf x$ that satisfy the power constraint \eqref{eq:power:constraint}.
Let 
\begin{equation}
\wrnd=\cinone+j\nlprm|\cinone|^2\cinone.\label{mosalas}
\end{equation} 
 We have that
 \begin{align}
 	I\rlpo{\cinone;\coutone}&=h\rlpo{\coutone}-h\rlpo{\coutone\given\cinone}\label{44}\\
 	&=h\rlpo{\wrnd+\noise}-h\rlpo{\wrnd+\noise\given\cinone}\\
 	&=h\rlpo{\wrnd+\noise}-h\rlpo{\noise}.\label{eq:wbased:mi}
 \end{align}
 Using the entropy power inequality \cite[Sec. 2.2]{elgamal_2011_net} and the Gaussian entropy formula \cite[Th.~8.4.1]{cover_information2}, we conclude that
 \begin{align}
 	h\rlpo{\wrnd+\noise} &\geq \log\rlpo{2^{h\rlpo{\wrnd}}+2^{h\rlpo{\noise}}}\\
 	&=\log\rlpo{2^{h\rlpo{\wrnd}}+\pi e \pn}.\label{eq:wpn}
 \end{align}
 Substituting \eqref{eq:wpn} into \eqref{eq:wbased:mi}, and using again the Gaussian entropy formula \cite[Th.~8.4.1]{cover_information2}, we obtain
 \begin{align}
 	I\rlpo{\cinone;\coutone}\geq\log\rlpo{2^{h\rlpo{\wrnd}}+\pi e \pn}-\log\rlpo{\pi e \pn}\label{eq:i}.
 \end{align}
We take $\cinone$  circularly symmetric. It follows from \eqref{mosalas} that $\wrnd$ is also circularly symmetric. Using \cite[Eq.~(320)]{lapidoth_2003_tinf} to compute $h(\wrnd)$, we obtain   
\begin{align}
h\rlpo{\wrnd}= h\rlpo{|\wrnd|^2}+\log{\pi}.\label{eq:hw}
\end{align}   
 Substituting \eqref{eq:hw} into \eqref{eq:i}, we get
 \begin{align}\label{eq:mi:exp:hu}
 	I\rlpo{\cinone;\coutone}\geq\log\rlpo{\frac{2^{h\rlpo{|\wrnd|^2}}}{ e \pn}+1}.
 \end{align}

Next to evaluate the  right-hand side (RHS) of \eqref{eq:mi:exp:hu}, we choose the following distribution for the amplitude square   $\srnd=|\cinone|^2$ of $\cinone$:
\begin{equation}
	f_{\srnd}(s)=\alfprm\rlp{3\nlprm^2s^2+1}e^{-\lmdprm s},\ \ \ \ s\geq 0.\label{eq:sdist}
\end{equation} 
 The parameters  $\lmdprm>0$ and $\alfprm>0$ are chosen so that \eqref{eq:sdist} is a pdf and so that the power constraint \eqref{eq:power:constraint}  is satisfied. We prove in Appendix~\ref{sec:alf:and:lmd} that by choosing these two parameters so that
 \begin{IEEEeqnarray}{c}\label{eq:alf:tm}
 \alfprm =\frac{\lmdprm^3}{\lmdprm^2+6\nlprm^2}
 \end{IEEEeqnarray}
and so that \eqref{eq:lmd:tm} holds, both constraints are met.
In Appendix~\ref{sec:hu}, we then prove that
\begin{align}\label{eq:hu}
h\rlpo{|\wrnd|^2}=-\log{\alfprm}+\alfprm\rlp{\frac{1}{\lmdprm}+\frac{18\nlprm^2}{\lmdprm^3}}\log e.
\end{align}
Substituting \eqref{eq:hu} and \eqref{eq:alf:tm} into \eqref{eq:mi:exp:hu}, we obtain \eqref{eq:ch1:lb}.
 Although not necessary for the proof, in Appendix~\ref{sec:optdist}, we justify  the choice of the pdf in \eqref{eq:sdist} by showing that it  maximizes $h(\wrnd)$.

\subsection{Choosing $\alfprm$ and $\lmdprm$}\label{sec:alf:and:lmd}
 We choose the coefficients $\alfprm$ and $\lmdprm$ so that \eqref{eq:sdist} is a valid pdf and $\expe\rlbo{\srnd}\leq P$.  Note that
\begin{align}
\int\limits_{0}^{\infty}f_{\srnd}(s)\di s&=\int\limits_{0}^{\infty}\alfprm\rlp{3\nlprm^2s^2+1}e^{-\lmdprm s}\di s\\
&=\alfprm\frac{\lmdprm^2+6\nlprm^2}{\lmdprm^3}.\label{eq:pwr:cnst}
\end{align}
Therefore, choosing $\alfprm$ according to \eqref{eq:alf:tm}  guarantees that $f_{\srnd}\rlpo{s}$ integrates to $1$. We next compute $\expe\rlbo{\srnd}:$
\begin{align}
\expe\rlbo{\srnd}&=\int\limits_{0}^{\infty}sf_{\srnd}(s)\di s\\
&=\int\limits_{0}^{\infty}s\alfprm\rlp{3\nlprm^2s^2+1}e^{-\lmdprm s}\di s\\
&=\alfprm\rlp{\frac{18\nlprm^2}{\lmdprm^4}+\frac{1}{\lmdprm^2}}.\label{eq:pwchkr}
\end{align}
 Substituting \eqref{eq:alf:tm} into \eqref{eq:pwchkr}, we obtain
 \begin{align}
\expe\rlbo{s}&=\frac{18\nlprm^2+\lmdprm^2}{\lmdprm\rlp{\lmdprm^2+6\nlprm^2}}.\label{eq:pwchkr2}
 \end{align}
We  see now that imposing $\expe\rlbo{\srnd}\leq\po$ is equivalent to  \eqref{eq:powr:cons}. Observe that the RHS of \eqref{eq:pwchkr2} and the objective function on the RHS of \eqref{eq:ch1:lb} are  decreasing functions of $\lmdprm$. Therefore,  setting the RHS of \eqref{eq:pwchkr2} equal to $P$, which yields  \eqref{eq:lmd:tm}, maximizes the objective function in \eqref{eq:ch1:lb}. 
 
 Finally, we prove that \eqref{eq:lmd:tm} has a single positive root. We have
 \begin{IEEEeqnarray}{rCl}
 	 	f(\lmdprm)&=&\po \lmdprm^3-\lmdprm^2 +6\po\nlprm^2\lmdprm-18\nlprm^2\label{59}\\
 	 	&=&\rlp{\lmdprm^2+6\nlprm^2}\rlp{\po\lmdprm-1}-12\nlprm^2. \label{eq:prm:2}
 \end{IEEEeqnarray}
Note that $f(\lmdprm)\to \infty$ as $\lmdprm\to \infty$ and that the RHS of \eqref{eq:prm:2} is negative when $\lmdprm<1/P$. Furthermore, $f(\lmdprm)$ is monotonically increasing in the interval $[1/P,\infty)$. Indeed, when ${\lmdprm\geq 1/\po}$,
 \begin{IEEEeqnarray}{rCl}
 	\frac{\di}{\di \lmdprm}f(\lmdprm)&=&3\po \lmdprm^2-2\lmdprm +6\po\nlprm^2\\
 	&\geq& 3 \lmdprm -2\lmdprm +6\po\nlprm^2\label{eq:lprm:2}\\
 	&>&0.
 	\end{IEEEeqnarray}
This yields the desired result.

\subsection{Proof of \eqref{eq:hu}}\label{sec:hu}
To compute the differential entropy of $\urnd=|\wrnd|^2$, we first determine the pdf of $\urnd$. By definition,
\be\label{eq:u:and:s}
\urnd=\srnd+\nlprm^2\srnd^3.
\ee
Let now $\gfun\rlpo{x}=x+\nlprm^2x^3$. Since 
\begin{align}
\frac{\di}{\di x}\gfun\rlpo{x}&=1+3\nlprm^2x^2
\end{align}
we conclude that $\gfun\rlpo{x}$  is monotonically increasing for $x\geq0$. Hence, $g(x)$ is one-to-one when $x\geq 0$ and its inverse 
\begin{equation}
\qfun\rlpo{x}=\gfun^{-1}\rlpo{x}\label{eq:qfun:def}
\end{equation}
 is well defined. Thus, the pdf of $\urnd$ is given by \cite[Ch.~5]{papolis_PrSt}
\begin{align}
f_{\urnd}(t)&=\frac{f_{\srnd}\rlpo{\qfun\rlpo{t}}}{\gfun'\rlpo{\qfun\rlpo{t}}}\label{eq:upb:hx2}\\
&=\frac{\alfprm\rlpo{3\nlprm^2\qfun^2\rlpo{t}+1}e^{-\lmdprm\qfun\rlpo{t}}}{3\nlprm^2\qfun^2\rlpo{t}+1}\label{71}\\
&=\alfprm e^{-\lmdprm q(t)},\ \ \ t\geq 0.\label{72}
\end{align}
Here, \eqref{71} holds because of \eqref{eq:sdist}.
Using \eqref{72}, we can now compute $h\rlpo{\urnd}$ as
\begin{IEEEeqnarray}{r?rCl}
&h\rlpo{\urnd}&=&-\int\limits_{0}^{\infty} f_{\urnd}(t)\log\rlpo{f_{\urnd}(t)}\di t\\
&&=&-\log {\alfprm}+\lmdprm\alfprm \rlp{\log e} \int\limits_{0}^{\infty}\qfun\rlpo{t}e^{-\lmdprm\qfun\rlpo{t}}\di t \\
& &=&-\log {\alfprm}+\lmdprm\alfprm\rlp{\log e}{\int\limits_{0}^{\infty} re^{-\lmdprm r}\rlp{1+3\nlprm^2r^2}\di r }\label{68} \\
&&=&-\log {\alfprm}+\lmdprm\alfprm\rlp{\frac{1}{\lmdprm^2}+\frac{18\nlprm^2}{\lmdprm^4}}\log e
\end{IEEEeqnarray}
where in \eqref{68} we used the change of variables $r=\qfun\rlpo{t}$. This proves \eqref{eq:hu}.
\subsection{$f_{\srnd}(s)$ maximizes $h(\wrnd)$}\label{sec:optdist}
We shall prove  that the pdf $f_{\srnd}(s)=\alfprm\rlp{3\nlprm^2s^2+1}e^{-\lmdprm s},\ \ s\geq 0$, maximizes $h\rlpo{\wrnd}$. It follows  from \eqref{eq:hw} that to maximize $h\rlpo{\wrnd}$, we need to maximize $h\rlpo{\urnd}$. We assume that the power constraint is fulfilled with equality, i.e., that
\begin{IEEEeqnarray}{c}
	\expe\rlbo{\srnd}=\int_{0}^{\infty}sf_\srnd(s)\di s =P.
\end{IEEEeqnarray}
Using the change of variables $s=\qfun(t)$, where $q(t)$ was defined in \eqref{eq:qfun:def}, we obtain
\begin{IEEEeqnarray}{rCl}\label{eq:pw:t}
	\int_{0}^{\infty}\qfun(t)f_\srnd(\qfun(t))\qfun'(t)\di t =P.
\end{IEEEeqnarray}
Substituting \eqref{eq:upb:hx2} into \eqref{eq:pw:t} and using that $\qfun'(t)=1/\gfun'(q(t))$, we obtain
\begin{IEEEeqnarray}{rCl}\label{eq:pw:tt}
	\int_{0}^{\infty}\qfun(t)f_\urnd(t)\di t =P.
\end{IEEEeqnarray}
It follows now from \cite[Th.~12.1.1]{cover_information2}  that the pdf that maximizes $h(\urnd)$ is of the form $f_\urnd(t)=e^{\lambda_0+\lambda_1\qfun(t)},\ \ t\geq0$, where $\lambda_0$ and $\lambda_1$ need to be chosen so that \eqref{eq:pw:tt} is satisfied and $f_\urnd(t)$ integrates to one. Using \eqref{eq:upb:hx2}, we get
\begin{IEEEeqnarray}{rCl}
	f_\srnd(s)&=&f_\urnd(\gfun(s))\gfun'(s)\\
	&=& \rlpo{1+3\nlprm^2s^2}e^{\lambda_0+\lambda_1s},\ \ \ s\geq 0.
\end{IEEEeqnarray}
 By setting $\alfprm=e^{\lambda_0}$ and $\lambda=\lambda_1$, we obtain \eqref{eq:sdist}.
  
 \section{Proof of Theorem~\ref{thm:ub:x3}}\label{app:proof:thm:ub:x3}
Fix $\lambda\geq 0$. It follows from \eqref{capacity} and \eqref{eq:power:constraint} that
 \begin{align}\label{eqd1}
 	\ctwo(P)&\leq \sup\rlao{I(\cinone ; \coutone )+\lambda\rlp{1-\frac{\expe\lefto[|\cinone|^2\right]+\pn}{P+\pn}}}
 \end{align}
 where the supremum is over the set of  probability distributions that satisfy the power constraint \eqref{eq:power:constraint}. 
Next, we upper-bound the mutual information $I\rlpo{\cinone;\coutone}$ as 
\begin{align}
I\rlpo{\cinone;\coutone}&= h\rlpo{\coutone}-h\rlpo{\coutone\mid \cinone}\\
&= h\rlpo{\coutone}-h\rlpo{\noise}\\
&= h\rlpo{|\coutone|}+h\rlpo{\phase{\,\coutone}\mid |\coutone|}+\expe\rlpo{\log |\coutone|}-h\rlpo{\noise}\label{76}\\
&= h\rlpo{|\coutone|^2}+h\rlpo{\phase{\,\coutone}\mid |\coutone|}-\log 2-h\rlpo{\noise}\label{77}\\
&\leq h\rlpo{|\coutone|^2}+\log\rlp{\pi}-h\rlpo{\noise}\label{78}\\
&=h\rlpo{|\coutone|^2}-\log\rlpo{e\pn}\label{eq:ub:Ibyhv}
\end{align}
where in \eqref{76} we used  \cite[Lemma~6.16]{lapidoth_2003_tinf} and in \eqref{77} we used  \cite[Lemma~6.15]{lapidoth_2003_tinf}.
We fix now an arbitrary input pdf $f_\cinone(\cdot)$ that satisfies the power constraint and define the random variables $\vrnd=|\coutone|^2$ and $\wrnd=\cinone+j\nlprm|\cinone|^2\cinone$. 
 Next, we shall obtain an upper bound on  $h\rlpo{\vrnd}$ that is valid for all  $f_\cinone(\cdot)$. 
Let
\be\label{eq:distv}
\tilde{f}_{\vrnd}(v)=\betprm e^{-\gmmprm\qfun\rlpo{v}},\ \ \ v\geq 0
\ee 
for some parameters $\betprm>0$ and $\gmmprm>0$. The function $\qfun(\cdot)$ is defined in \eqref{eq:qfun:def}. 
We next choose~$\betprm$ so that $\tilde{f}_{\vrnd}(v)$ is a valid pdf. To do so,   we set $z=\qfun\rlpo{v}$, which implies that $\gfun(z)=v$, and that
\be
\rlp{1+3\nlprm^2z^2}\di z=\di v.
\ee
Therefore, integrating $\tilde{f}(v)$ in \eqref{eq:distv}, we obtain
\begin{align}
	\int\limits_{0}^{\infty}\betprm e^{-\gmmprm\qfun\rlpo{v}}\di v&=\betprm\int\limits_{0}^{\infty} e^{-\gmmprm z}\rlp{1+3\nlprm^2z^2}\di z \\
	&=\betprm\rlp{\frac{\gmmprm^2+6\nlprm^2}{\gmmprm^3}}.\label{eq:betprm:proof}
\end{align}
We see from \eqref{eq:betprm:proof} that the choice
\begin{IEEEeqnarray}{c}
\betprm=\frac{\gmmprm^3}{\gmmprm^2+6\nlprm^2}\label{eq:betprm:def}
\end{IEEEeqnarray}
  makes $\tilde{f}_{\vrnd}(v)$  a valid pdf.
 Using the definition of the relative entropy, we have 
%
\begin{align}\label{85}
D\rlpo{f_{\mf v}(v) \, |\!|\, \tilde f_{\mf v}(v)}&=\int_{-\infty}^{+\infty}f_{\mf v}(v)\log\rlpo{\frac{f_{\mf v}(v)}{\tilde f_{\mf v}(v)}} \di v\\
&=-h(\mf v)-\expe\rlbo{\log\rlpo{\tilde f_{\mf v}(v)}}.
\end{align}
Since the relative entropy is nonnegative \cite[Thm.~8.6.1]{cover_information2}, we obtain
\begin{align}\label{87}
	h\rlpo{\mf v}&\leq -\expe_{\mf v}\rlpo{\log\rlpo{\tilde{f}_{\mf v}(\mf v)}}\\
	&=-\log{\betprm}+\gmmprm \expe\rlbo{\qfun\rlpo{\vrnd}}\log e. \label{e2}
\end{align}

Substituting \eqref{eq:ub:Ibyhv} and \eqref{e2} into \eqref{eqd1}, we obtain
\begin{align}
 	\ctwo(P)&\leq -\log\rlpo{e\pn}-\log{\betprm}+\lambda+ \sup\rlao{\gmmprm \expe\rlbo{\qfun\rlpo{\vrnd}}\log e-\lambda\frac{\expe\lefto[|\cinone|^2\right]+\pn}{P+\pn}}\\
 	&\leq -\log\rlpo{e\pn}-\log{\betprm}+\lambda+ \max_{s>0}\rlao{\gmmprm \expe\rlbo{\qfun\rlpo{\vrnd}\mid |\cinone|^2=s}\log e-\lambda\frac{s+\pn}{P+\pn}}.\label{le}
\end{align}
The final upper bound  \eqref{eq:thm:ub:x3:hp} is obtained  by minimizing   \eqref{le}  over all $\lambda\geq 0$ and $\gmmprm\geq 0$.

\section{Proof of Theorem \ref{thm:x3:ubn:lp}}\label{app:proof:ub:x3:lp}
It follows from \eqref{e2} and \eqref{eq:ub:Ibyhv} that
\begin{align}\label{III}
	I\rlpo{\cinone;\coutone}&\leq -\log\rlpo{\betprm}+ \gmmprm \expe\rlbo{\qfun\rlpo{\vrnd}}\log e-\log\rlpo{e\pn}
\end{align}
where $\vrnd=|\mf y|^2$. Moreover,
\begin{align}
\expe\rlbo{\qfun\rlpo{\vrnd}}
&= \expe\rlbo{\qfun\rlpo{|\wrnd+\noise|^2}} \\
&= \expe\rlbo{\qfun\rlpo{|\wrnd|^2+|\noise|^2+2\red\rlpo{\wrnd\noise^*}}}.\label{eq:hv:upb:Evq}
\end{align}
Next, we analyze the function $\qfun(x)$. We have
\begin{align}
\qfun'\rlpo{x}&=\frac{1}{\gfun'\rlpo{\qfun\rlpo{x}}}\\
&=\frac{1}{1+3\nlprm^2\qfun^2\rlpo{x}}.\label{eq:qfun:fstdriv}
\end{align}
Furthermore,
\begin{align}\label{eq:qfun:secdriv}
\qfun''\rlpo{x}&=-\frac{6\nlprm^2\qfun\rlpo{x}}{\rlpo{1+3\nlprm^2\qfun^2\rlpo{x}}^3}\leq 0.
\end{align}
 Therefore,  $\qfun\rlpo{x}$ is a nonnegative concave function on $[0,\infty)$. Thus, for every real numbers $x\geq 0$ and $y\geq -x$,
\begin{align}
\qfun\rlpo{x+y}&\leq \qfun\rlpo{x}+\qfun'(x)y\label{eq:qfun:upb:line}.
\end{align}
 Using \eqref{eq:qfun:upb:line} in \eqref{eq:hv:upb:Evq},  with $x=|\wrnd|^2+|\noise|^2$ and $y=2\red\rlpo{\wrnd\noise^*}$, we get
\begin{align}
	\expe\rlbo{\qfun\rlpo{\vrnd}}&\leq \expe\rlbo{\qfun\rlpo{|\wrnd|^2+|\noise|^2}+2\qfun'\rlpo{|\wrnd|^2+|\noise|^2}\red\rlpo{\wrnd\noise^*}}.\label{eq:102}
	\end{align}
	 Using \eqref{eq:qfun:upb:line} once more  with $x=|\wrnd|^2$ and $y=|\noise|^2$, we obtain
\begin{align}
\expe\rlbo{\qfun\rlpo{\vrnd}}&\leq \expe\rlbo{\qfun\rlpo{|\wrnd|^2}+\qfun'\rlpo{|\wrnd|^2}|\noise|^2+2\qfun'\rlpo{|\wrnd|^2+|\noise|^2}\red\rlpo{\wrnd\noise^*}}\\
&= {\expe\rlbo{\qfun\rlpo{|\wrnd|^2}}+\pn \expe\rlbo{\qfun'\rlpo{|\wrnd|^2}}+2 \expe\rlbo{\qfun'\rlpo{|\wrnd|^2+|\noise|^2}\red\rlpo{\wrnd\noise^*}}}.\label{eq:ubn:hv:expe}
\end{align}
 We shall now bound each  expectation in \eqref{eq:ubn:hv:expe} separately. Since $|\wrnd|^2=\gfun\rlpo{|\cinone|^2}$, we have that
 \begin{align}
 \expe\rlbo{\qfun\rlpo{|\wrnd|^2}}&= \expe\rlbo{|\cinone|^2}\\
 &\leq\po\label{eq:upb:1st}
 \end{align}
 where the last inequality follows from \eqref{eq:power:constraint}.
 It also follows from \eqref{eq:qfun:fstdriv} that
 \begin{align}
\qfun'(|\wrnd|^2)\leq 1.\label{eq:upb:2nd}
 \end{align}
Furthermore, \eqref{eq:qfun:fstdriv} and \eqref{eq:qfun:secdriv} imply that the function $\qfun'\rlpo{x}$ is positive and decreasing  in the interval $x\geq 0$. Therefore,
\begin{align}
 \expe\rlbo{\qfun'\rlpo{|\wrnd|^2+|\noise|^2}\red\rlpo{\wrnd\noise^*}}&\leq \expe\rlbo{\qfun'\rlpo{|\wrnd|^2+|\noise|^2}|\red({\wrnd\noise^*})|}\label{eq:110}\\
 &\leq \expe\rlbo{\qfun'\rlpo{|\wrnd|^2}|\red({\wrnd\noise^*})|}\\
  &\leq \expe\rlbo{\qfun'\rlpo{|\wrnd|^2}|\wrnd|\cdot|\noise|}\\
    &= \frac{\sqrt{\pi\pn}}{2}\expe\rlbo{\qfun'\rlpo{|\wrnd|^2}|\wrnd|}\label{eq:En}\\
        &\leq \frac{\sqrt{\pi\pn}}{2}\max\limits_{t\geq0}\rlao{t\qfun'\rlpo{t^2}}\\
&= \frac{\sqrt{\pi\pn}}{2}\max\limits_{t\geq0}\rlao{\frac{t}{1+3\nlprm^2\qfun^2\rlpo{t^2}}}\label{eq:upb:3rd1}
\end{align}
where \eqref{eq:En} holds because $\expe\rlbo{|\noise|}=\sqrt{\pi\pn/4}$  and the last equality follows from \eqref{eq:qfun:fstdriv}.
To calculate the maximum in \eqref{eq:upb:3rd1}, we use the change of variables $t^2=x+\nlprm^2x^3$ to obtain
\begin{IEEEeqnarray}{rCl}
	\max\limits_{t\geq0}\rlao{\frac{t}{1+3\nlprm^2\qfun^2\rlpo{t^2}}}&=& \max\limits_{x\geq0}\rlao{\frac{\sqrt{x+\nlprm^2x^3}}{1+3\nlprm^2x^2}}\label{eq:y1}\\
	&=&\frac{1}{12^{3/8}\sqrt{\rlp{\sqrt{3}-1}\nlprm}}\label{eq:obt:sol}
\end{IEEEeqnarray}
where the last step follows by some standard algebraic manipulations that involve finding the  roots of the derivative of the objective function on the RHS of \eqref{eq:y1}.
 Substituting \eqref{eq:obt:sol} into \eqref{eq:upb:3rd1}, we obtain
 \begin{align}
 	\expe\rlbo{\qfun'\rlpo{|\wrnd|^2+|\noise|^2}\red\rlpo{\wrnd\noise^*}}&\leq \frac{\sqrt{\pi\pn}}{2\times12^{3/8}\sqrt{\rlp{\sqrt{3}-1}\nlprm}}.\label{eq:upb:3rd}
 \end{align}

 Substituting \eqref{eq:upb:1st}, \eqref{eq:upb:2nd}, and \eqref{eq:upb:3rd} into \eqref{eq:ubn:hv:expe}, and the result into \eqref{III}, we obtain 
 \begin{align}\label{eq:ub:mi:hp}
 I\rlpo{\cinone;\coutone}& \leq-\log{\betprm}+\gmmprm\rlp{\log e}\rla{ \po+\pn +\frac{\sqrt{\pi\pn}}{12^{3/8}\sqrt{\rlp{\sqrt{3}-1}\nlprm}}}-\log(e\pn).
 \end{align}
Finally, we obtain \eqref{eq:thm:ub:x3:hp} by substituting \eqref{eq:betprm:def} into \eqref{eq:ub:mi:hp}. Since the upper bound \eqref{eq:ub:mi:hp} on mutual information holds for every input distribution that satisfies the power constraint, it is also an upper bound on  capacity for every $\gmmprm>0$.
To find the optimal $\gmmprm$, we need to minimize
\begin{align}
\log\rlpo{\frac{\gmmprm^2+6\nlprm^2}{\gmmprm^3}}+\gmmprm\rlp{\po+\cprm}\log e =\log\rlpo{\frac{\exp\rlpo{\gmmprm\rlpo{\po+\cprm}}\rlp{\gmmprm^2+6\nlprm^2}}{\gmmprm^3}}\label{eq:optgmm:log}
\end{align}
where $\cprm$ was defined in  \eqref{eq:gmmprm:def}. 
Observe now that the function inside logarithm on the RHS of \eqref{eq:optgmm:log} goes to infinity when $\gmmprm\to 0$ and when $\gmmprm\to \infty$. 
Therefore, since this function is positive, it must have a minimum in the interval $[0 , \infty)$. To find this minimum, we set its derivative equal to zero and get \eqref{eq:gmmprm:def}. Note finally that since \eqref{eq:lmd:tm} has exactly one real root, which was proved in Appendix~\ref{sec:alf:and:lmd},   \eqref{eq:gmmprm:def} also has exactly one real root.

 \section{Proof of Theorem \ref{thm:ub:nl:ssfm}}\label{app:proof:ub:nl:ssfm}

The proof uses  similar steps as in  \cite[Sec.~III-C]{durisi2012capacity}. 
We upper-bound the mutual information between the $\mf x$ and $\mf y$ expressed in  polar coordinates as 
\begin{align}\label{eq1}
I(\mf x ; \mf y)&=	I(\mf r_0,\bm \theta_0 ; \mf r, \bm \theta )\\
	&=I(\mf r_0,\bm \theta_0;\mf r)+I(\mf r_0 ,\bm \theta_0;\bm \theta\mid \mf r)\\
	&=h(\mf r)-h(\mf r\mid \mf r_0,\bm \theta_0)+h(\bm \theta\mid \mf r)-h(\bm \theta\mid \mf r,\mf r_0,\bm \theta_0)\\
	&\leq h(\mf r^2) - \expe[\log(\mf r)] - \log 2-h(\mf r\mid \mf r_0,\bm \theta_0)+\log(2\pi)-h(\bm \theta\mid \mf r,\mf r_0,\bm \theta_0).\label{156.5}
\end{align}
In \eqref{156.5} we used \cite[Eq.~(317)]{lapidoth_2003_tinf} and that $h(\bm \theta\mid \mf r)\leq \log(2\pi)$. Let now $\tilde f_{\mf r^2}(\cdot)$ denote an arbitrary pdf for $\mf r^2$. 
Following the same calculations as in \eqref{85}--\eqref{87}, we obtain
\begin{align}
	h(\mf r^2)\leq -\expe_{\mf r^2}\lefto[\log(\tilde f_{\mf r^2}(\mf r^2))\right].\label{157}
\end{align}
We shall take $\tilde f_{\mf r^2}(\cdot)$ to be a Gamma distribution with parameters $\alpha>0$ and  $\beta=(\po+\pn)/\alpha$, i.e.,
\begin{align}
	\tilde f_{\mf r^2}(z)=\frac{z^{\alpha-1}e^{-z/\beta}}{\beta^\alpha \Gamma(\alpha)},\qquad z\geq 0.\label{158}
\end{align}
Here, $\Gamma(\cdot)$ denotes the Gamma function.
 Substituting \eqref{158} into \eqref{157}, we obtain
\begin{align}
	&\expe[\log(\tilde f_{\mf r^2}(\mf r^2))]\\
	&=2(\alpha-1)\expe\lefto[\log(\mf r)\right]-{\alpha}\frac{\expe\lefto[\mf r_0^2\right]+\pn}{P+\pn}\log e-\alpha\log\rlpo{\frac{\po+\pn}{\alpha}}-\log(\Gamma(\alpha)).\label{161}
\end{align}
It follows from \eqref{cd:pdf} that the random variables $\mf r$ and $\bm \theta_0$ are conditionally independent given $\mf r_0$. Therefore, 
\begin{equation}
h(\mf r\mid \mf r_0,\bm \theta_0)=h(\mf r\mid \mf r_0).\label{162}
\end{equation}
 Next, we study the term $h(\bm \theta\mid \mf r,\mf r_0,\bm \theta_0)$ in \eqref{156.5}.
  From Bayes' theorem 
 and \eqref{cd:pdf} it follows that
  for every  $\theta'\in [0,2\pi)$
\begin{equation}
f_{\bm \theta\mid \mf r, \mf r_0, \bm \theta_0}(\theta|r, r_0,\theta_0)=f_{\bm \theta\mid \mf r, \mf r_0, \bm \theta_0}(\theta-\theta'|r, r_0,\theta_0-\theta').
\end{equation}
Therefore, 
\begin{align}
	h(\bm \theta\mid \mf r,\mf r_0, \bm \theta_0)&=\int_0^{2\pi}f_{\bm \theta_0}(\theta_0) h(\bm \theta\mid \mf r, \mf r_0, \bm \theta_0=\theta_0)\di \theta_0\\
	&=\int_0^{2\pi}f_{\bm \theta_0}(\theta_0) h(\bm \theta-\theta_0\mid\mf r, \mf r_0, \bm \theta_0=0)\di \theta_0\label{156}\\
	&=\int_0^{2\pi}f_{\bm \theta_0}(\theta_0) h(\bm \theta\mid \mf r_0, \bm \theta_0=0)\di \theta_0\label{158s}\\
	&=h(\bm \theta\mid \mf r, \mf r_0, \bm \theta_0=0).\label{166.5}
\end{align}
Here, \eqref{158s} follows because  differential entropy is invariant to translations \cite[Th.~8.6.3]{cover_information2}. 
Substituting \eqref{161}, \eqref{157}, \eqref{162}, and \eqref{166.5} into \eqref{156.5}, we obtain
\begin{align}
	I(\mf r_0,\bm \theta_0 ; \mf r, \bm \theta )&\leq\alpha\log\rlpo{\frac{\po+\pn}{\alpha}}+\log\lefto(\Gamma(\alpha)\right)+\log(\pi)+{\alpha}\frac{\expe\lefto[\mf r_0^2\right]+\pn}{P+\pn}\log e\nonumber\\
	&\qquad+(1-2\alpha)\expe\lefto[\log(\mf r)\right]-h\lefto(\mf r\mid \mf r_0\right) -h(\bm \theta\mid \mf r, \mf r_0, \bm \theta_0=0).\label{mi_up}
\end{align}
Fix $\lambda\geq 0$. We next upper bound $\ctwo{}$ using \eqref{mi_up} as 
\begin{align}
	\ctwo(P)&\leq \sup\rlao{I(\mf r_0,\bm \theta_0 ; \mf r, \bm \theta )+\lambda\rlp{1-\frac{\expe[\mf r_0^2]+\pn}{P+\pn}}}\\
	&\leq \alpha\log\rlpo{\frac{\po+\pn}{\alpha}}+\log\lefto(\Gamma(\alpha)\right)+\log(\pi)+\lambda\nonumber\\
	&\qquad+\sup\bigg\{({\alpha\log e}-\lambda)\frac{\expe\lefto[\mf r_0^2\right]+\pn}{P+\pn}
	+(1-2\alpha)\expe\lefto[\log(\mf r)\right]-h\lefto(\mf r\mid \mf r_0\right)\nonumber\\
	& \qquad-h(\bm \theta\mid \mf r, \mf r_0, \bm \theta_0=0)\bigg\}\label{171}
\end{align}
where the supremum is over the set of input probability distributions that satisfy \eqref{eq:power:constraint}. We complete the proof  by noting that the supremum in \eqref{171} is less or equal to $\max_{r_0>0}\{g_{\lambda,\alpha}(r_0, P)\}$, where $g_{\lambda,\alpha}(r_0, P)$ is defined in \eqref{gdef}.

\ifCLASSOPTIONcaptionsoff
  \newpage
\fi


 \end{document}